\documentclass[11pt]{article}
\usepackage[margin=1in]{geometry}
\usepackage[utf8]{inputenc}
\usepackage[english]{babel}
\usepackage{amsmath,amssymb,bbm,amsthm,xspace,url,hyperref}
\usepackage{times,mathptmx}
\usepackage{cleveref}
\newif\ifpazo\pazofalse

\def\clap#1{\hbox to 0pt{\hss#1\hss}}

\newcommand{\xqedhere}[2]{\rlap{\hbox to#1{\hfil\llap{\ensuremath{#2}}}}}

\ifpazo \newcommand{\indic}{\mathbb{1}} 
\else \newcommand{\indic}{\mathbbm{1}} \fi
\newcommand{\EE}{\mathbb{E}\,}

\newcommand{\cB}{\mathcal{B}}

\newcommand{\cS}{\mathcal{S}}
\newcommand{\cT}{\mathcal{T}}

\newcommand{\cX}{\mathcal{X}}

\newcommand{\bx}{\mathbf{x}}

\newcommand{\tO}{\widetilde{O}}
\newcommand{\tOmega}{\widetilde{\Omega}}
\newcommand{\tTheta}{\widetilde{\Theta}}

\renewcommand{\ge}{\geqslant}
\renewcommand{\le}{\leqslant}

\renewcommand{\b}{\{0,1\}}

\newcommand{\ceil}[1]{\lceil{#1}\rceil}
\newcommand{\floor}[1]{\lfloor{#1}\rfloor}

\newcommand{\relent}[2]{\mathrm{D}(#1\|#2)}

\newtheorem{theorem}{Theorem}

\newtheorem{lemma}[theorem]{Lemma}

\newtheorem{prop}[theorem]{Proposition}
\theoremstyle{definition}

\theoremstyle{remark}

\DeclareMathOperator{\Bin}{Bin}
\DeclareMathOperator{\Bern}{Bern}

\newcommand{\mem}{\textsc{memory}\xspace}

\title{Time-Space Tradeoffs for the Memory Game}
\author{Amit Chakrabarti\thanks{Department of Computer Science, Dartmouth College, Hanover, NH. Supported in part by NSF under award CCF-1650992.}
\and Yining Chen\thanks{Department of Computer Science, Dartmouth College, Hanover, NH.}}
\date{}

\begin{document}

\maketitle
\thispagestyle{empty}

\begin{abstract}
A single-player game of \textsc{memory} is played with $n$ distinct pairs of cards, with the cards in each pair bearing identical pictures. The cards are laid face-down. A move consists of revealing two cards, chosen adaptively. If these cards match, i.e., they bear the same picture, they are removed from play; otherwise, they are turned back to face down. The object of the game is to clear all cards while minimizing the number of moves. Past works have thoroughly studied the expected number of moves required, assuming optimal play by a player has that has {\em perfect} memory. In this work, we study the \textsc{memory} game in a {\em space-bounded} setting.

We prove two time-space tradeoff lower bounds on algorithms (strategies for the player) that clear all cards in $T$ moves while using at most $S$ bits of memory. First, in a simple model where the pictures on the cards may only be compared for equality, we prove that $ST = \Omega(n^2 \log n)$. This is tight: it is easy to achieve $ST = O(n^2 \log n)$ essentially everywhere on this tradeoff curve. Second, in a more general model that allows arbitrary computations, we prove that $ST^2 = \Omega(n^3)$. We prove this latter tradeoff by modeling strategies as branching programs and extending a classic counting argument of Borodin and Cook with a novel probabilistic argument. We conjecture that the stronger tradeoff $ST = \widetilde{\Omega}(n^2)$ in fact holds even in this general model.

\bigskip


\end{abstract}

\newpage
\setcounter{page}{1}


\section{Introduction}

The popular children's card game ``Memory'' (also known as ``Concentration'') is played with a deck of $n$ pairs of picture cards: the two cards in each pair have the same picture, while two cards taken from two distinct pairs have distinct pictures. The game starts with these cards facing down on a table. A {\em move} in the game consists of flipping up one card to reveal it, and then flipping up a second card: if the two revealed cards are from the same pair---i.e., they {\em match}---they are removed from the table, otherwise they are flipped back to face down. The game ends when all cards have been removed. From these basic rules, one can formulate a two-player or multi-player game, but this paper is concerned with the one-player ``Solitaire Memory'' game, which we shall simply call \mem in this paper. The goal of \mem is to remove all cards as quickly as possible, i.e., using a minimum number of moves.

Versions of this game have attracted some attention from researchers in the recent past. Alfthan~\cite{Alfthan07} studied strategies for two-player versions of the game. Foerster and Wattenhofer~\cite{FoersterW13} studied the solitaire game and derived the first nontrivial upper and lower bounds on the expected number, $T(n)$, of moves required by an optimal strategy. Subsequently, Velleman and Warrington~\cite{VellemanW13} proved the very tight result that $T(n) = (3 - 2\ln 2)n + 7/8 - 2\ln 2 + o(1) \approx 1.61n$. However, all of these works assume that the player(s) have {\em perfect} memory. Arguably, what makes the children's game Memory interesting is that player(s) will find it hard to remember everything they have learned from previous moves. Thus, from a computational complexity perspective, the most interesting question about \mem is: what is the minimum expected number of moves required by a player whose memory capacity is at most $S$ bits?

\subsection{Our Results and Techniques}

Define an $(S,T)$-strategy to be one where the player flips at most $T = T(n)$ cards, while using at most $S = S(n)$ bits of memory throughout. Notice that we are measuring ``time'' in terms of cards flipped (rather than moves made): this will simplify the presentation and is accurate enough for our purposes, since we shall prove asymptotic results. We shall soon make this definition more formal. For now, we note that we can readily obtain deterministic $(S,T)$-strategies with $ST = \tO(n^2)$, for essentially all $S$ between $\tTheta(1)$ and $\tTheta(n)$.\footnote{The $\tO(\cdot)$ and $\tTheta(.)$ notations ignore factors polylogarithmic in $n$; in this particular instance, the factor ignored is $O(\log n)$.} We assume that the picture on each card can be represented using $\tO(1)$ bits.

The main goal of this work is to study corresponding lower bounds, i.e., time-space tradeoffs. We prove two such tradeoffs. The first applies to a simple model in which the player is deterministic and can do only one thing with the pictures on the cards, namely, compare them for equality. Under this restriction, we show that $ST = \tOmega(n^2)$ is necessary; see \Cref{thm:blind}. We do this by analyzing an explicit adversarial strategy, using concepts from matching theory. The aforementioned upper bound $ST = \tO(n^2)$ is achieved by a strategy that does operate in this simple model, so this lower bound is tight.

Our second tradeoff (and main result) applies to a fully general model, where the player may use randomization and may treat the picture on each card as an integer and perform arbitrary computations with these integers.
We prove that every $(S,T)$-strategy must satisfy $ST^2 = \Omega(n^3)$; see \Cref{thm:main} in \Cref{sec:lb-setup}. We do this by modeling a strategy for \mem as a general branching program (BP): there are $2n$ input variables corresponding to the $2n$ cards; each node in the BP reads a variable (flips a card); and each edge in the BP has the potential to produce output (declare some pairs of cards as matched). To prove the appropriate time-space tradeoff for BPs that successfully play \mem, we extend a classic counting argument due to Borodin and Cook~\cite{BorodinC82} with a novel probabilistic argument of our own.

\subsection{Related Work}
\label{RelatedWork}

We have already discussed past work on the Memory game itself~\cite{Alfthan07,FoersterW13,VellemanW13}. As those works do not deal with a space-bounded computational setting, at the technical level they are largely unconnected to our work. The relevant related work is mostly in the area of time-space tradeoffs for branching programs. For a detailed overview of this area and an excellent exposition of many key results, we refer the reader to Chapter 10 of the textbook by Savage~\cite{Savage:1997:MCE:522806}. Our work was motivated in part by a goal of furthering our understanding of space-limited algorithms for computing matchings in graphs: \mem is a toy version of this much richer problem. Thus, the recent and growing body of work on streaming algorithms and lower bounds for computing matchings in graphs is also relevant: see, e.g., Assadi et al.~\cite{AssadiKL17}, Crouch and Stubbs~\cite{CrouchS14}, and the references therein. We discuss this motivation further at the end of the paper, in Section~\ref{sec:concl}.

In a highly influential work, Borodin and Cook~\cite{BorodinC82} studied the \textsc{sorting} problem: given integers $x_1, x_2 \dots, x_n \in [n^3]$, output the $x_i$ values in sorted order.\footnote{The notation $[N]$ denotes the set $\{1,2,\ldots,N\}$.}
They devised a counting-based method for deriving time-space lower bounds for branching programs, using which they proved the tradeoff $ST = \Omega(n^2/\log{n})$ for \textsc{sorting}. Their method's essential feature is that it divides time into ``stages'' and applies a counting argument to each stage to argue that, if the branching program is ``too small,'' then none of the stages can make enough progress on sufficiently many inputs. We shall see this overall scheme in the proof of our main theorem. The Borodin-Cook method has been applied to many other problems, including matrix multiplication over a finite field~\cite{Yesha84}, generalized string matching~\cite{Abrahamson87}, Boolean matrix multiplication~\cite{Abrahamson90}, calculating universal hash functions~\cite{MansourNT93}, and \textsc{unique-elements}~\cite{Beame91}.

 The Borodin-Cook method is described in detail in Savage's book~\cite{Savage:1997:MCE:522806}; in particular, Theorem 10.11.1 in that book gives a very general version of the method. In the terminology of that theorem, the conditions under which this method applies are captured by a property of functions called $(\phi, \lambda, \mu, \nu, \tau)$-distinguishability. The function that describes the desired output of the game \mem is not $(\phi, \lambda, \mu, \nu, \tau)$-distinguishable for any positive constant $\nu$, so this method does not apply as is to our problem.

It is worth comparing the quality of our branching program lower bound to known results on \textsc{sorting} and \textsc{unique-elements}. When the ``pictures'' on the cards are just the integers in $[R]$, \mem can be thought of as a special case of \textsc{sorting} where only a subset of $[R]^{2n}$ consisting of $n$ equal pairs are valid inputs. An algorithm for \textsc{sorting} can be used to output all $n$ pairs in order of increasing variable values, effectively generating the output for \mem. In \textsc{unique-elements}, the input consists of $n$ integers $x_1, x_2, \ldots, x_n \in [n]$ and the desired output is a list of all $i$ such that the value $x_i$ appears exactly once in the input. Consider the following closely related problem that we call \textsc{unique-pairs}: given an input of $2n$ integers in $[n]$, output all pairs $(i,j)$ with $i < j$ such that $x_i = x_j$ and no other $k \in [2n]$ satisfies $x_i = x_j = x_k$. Then, with minor modifications, the analysis that Beame gives for \textsc{unique-elements}~\cite{Beame91} still applies, so this variant has $ST = \Omega(n^2)$. (See Appendix~\ref{unique2} for some details of how to modify Beame's proof.) So playing \mem is an easier task than both \textsc{sorting} and \textsc{unique-elements}, and thus it is harder to prove a lower bound. Indeed, so far we are only able to show $ST^2 = \Omega(n^3)$ for our problem, instead of the stronger $ST = \Omega(n^2)$ bound known for these related problems.

Our other lower bound for \mem, which does give an optimal tradeoff of $ST = \Omega(n^2 \log n)$, applies in a weaker model where cards may only be compared for equality. Its proof is based on a direct adversarial argument in the style of classic lower bounds for deterministic query complexity; see, e.g., Chapter~5 of the textbook by Du and Ko~\cite{DuKo-book}.


\section{Preliminaries} \label{sec:prelim}

Without loss of generality, we may assume that the $2n$ cards in a game of \mem are laid down on the table linearly. We model our input as a $(2n)$-tuple of variables $\bx = (x_1,\ldots,x_{2n})$, where $x_i$ is the ``picture'' on the $i$th card. We think of each picture as an integer in $[R] := \{1,\ldots,R\}$. We model the act of flipping the cards as reading the values of these variables. Per the rules of \mem, the only {\em valid} tuples $\bx$ are those where exactly $n$ distinct values occur exactly twice each. We let $\cX$ denote the set of valid inputs. Given $\bx \in \cX$, a {\em match} in $\bx$ is a triple $(i,j,v)$ where $1 \le i < j \le 2n$, $1 \le v \le R$, and $x_i = x_j = v$. Clearly, every valid $\bx$ has exactly $n$ matches. The goal of \mem is to output all $n$ matches---this is how we model the act of removing all cards---under the promise that the input is valid.

It is reasonable to assume that $R = n^{O(1)}$: indeed, if $R$ were larger, we could simply hash each picture down to a $\ceil{3 \log n}$-bit string using a $2$-universal hash function and this would lead to a collision with probability at most $O(1/n)$. We are requiring the output to specify not just the matching pairs of cards $(i,j)$ but also the integer values shown on those cards. This, too, is reasonable because it costs at most $n$ additional variable reads to satisfy this requirement, and at least $n$ variables must be read by any correct strategy.

For the rest of this paper, we shall study computational complexity in worst-case settings for inputs. More precisely, an $(S,T)$-strategy is required to use at most $S$ bits of memory. In the deterministic case, the strategy must always terminate with a correct output after reading at most $T$ variables. In the randomized case, it must always terminate with a correct output and the expected number of variables it reads must be at most $T$; thus, this is a Las Vegas notion of randomization.

We proceed to formally define two computational models for playing \mem using limited memory (space). We then present a basic, and easy to prove, upper bound that applies to the weaker of these models. We shall eventually prove a tight lower bound in this weaker model and a non-tight lower bound in the stronger model.

\subsection{Computational Models}

\paragraph{Branching Programs.~} Our strongest model is that of {\em general branching programs} on the variables $x_1, \ldots, x_{2n}$. Recall that each variable takes values in $[R]$. An $R$-way branching program (BP) is a directed acyclic graph where each node is labeled with one of the variables $x_i$; there are exactly $R$ out-edges from each non-sink node, labeled with the $R$ different possible values for $x_i$; an edge is additionally annotated with zero or more outputs; and there is a single {\em source node}. Given an input $\bx$, the execution of the BP on $\bx$ starts at the source node. At each time step, the algorithm reads the variable $x_i$ labeling the {\em current node} and, based on its value, it branches to one of the $R$ successors of the node. In the process, it produces all the outputs that annotate the edge traversed. If the algorithm moves to a {\em sink node}, then it halts.

We can model an $(S,T)$-strategy for \mem as an $R$-way {\em layered} branching program $\cB$. This is simply a branching program whose nodes are arranged in layers, numbered from $0$ to $T$, with each edge going from one layer to the next. The source node is in layer $0$ and every node in layer $T$ is a sink. Further, since the strategy is limited to $S$ bits of working memory (space), each layer has at most $2^S$ nodes. It is conventional to say that $\cB$ has length $T$ and width at most $2^S$.

On every input $\bx = (x_1,\ldots,x_{2n}) \in [R]^{2n}$, running $\cB$ on $\bx$ will trace a source-to-sink path through its nodes, outputting some triples in the process. Correctness requires that for every valid input $\bx \in \cX$, exactly $n$ outputs are produced, namely, the $n$ matches in $\bx$.

\paragraph{The Blind Player Model.~} Since the two main actions in \mem are remembering some cards and comparing two cards for equality, the game can also be studied in a weaker model where such equality comparisons are the {\em only} thing one is allowed to do with the pictures on the cards. The BP model above allows a strategy to perform arbitrary computations (e.g., arithmetic) with the integers written on the cards. In contrast, in a {\em blind player model}, a strategy cannot see the pictures on the cards directly, but can only know whether or not two such pictures are equal. A strategy can, however, manipulate other relevant integers, such as the index positions at which particular cards appear.

More formally, an blind player $(S,T)$-strategy maintains a {\em working set} $\cS \subseteq [2n]$: its elements are the indices of all cards ``remembered'' by the strategy. The set $\cS$ is initially empty, and such that $|\cS| \log n \le S$ at all times, so that $\cS$ fits in $S$ bits of memory. In each time step, the strategy {\em examines} a card. Suppose it examines the $i$th card: then an oracle instantaneously informs it whether $x_i = x_j$, for each $j \in \cS$. Based on this information, the strategy outputs zero or more values and updates $\cS$ to a subset of $\cS \cup \{i\}$.

\subsection{A Basic Upper Bound} \label{sec:ub}

We proceed to analyze a natural strategy, implementable in the blind player model, that sets the benchmark upper bound $ST = O(n^2 \log n)$.

\begin{prop} \label{prop:ub}
  In the blind player model described above, for all $S$ with $\log n \le S \le n\log n$, \mem has an $(S,T)$-strategy where $ST = O(n^2 \log n)$.
\end{prop}
\begin{proof}
  Let $s := \floor{S/\log n}$ be the number of cards that will fit in memory. The algorithm proceeds in {\em passes} and each pass begins by clearing the working set to $\varnothing$. In pass $i$, the player reads $x_{(i-1)s+1}, \ldots, x_{is}$ and stores them all in memory, outputting any matches found. The player then reads $x_{is+1}, \ldots, x_{2n}$ in order, forgetting each card after it is read, and outputting any matches found between these cards and the ones stored in memory. Clearly, every match in the input is eventually found and each pass takes $O(n)$ time.
  
  This algorithm makes at most $\ceil{2n/s}$ passes, for a total time complexity of $T = O((2n/s)\cdot n)$. This gives $ST = O(n^2 \log n)$, as desired.
\end{proof}

Note that our model does not account for the space used by a strategy to keep track of its progress, e.g., through loop indices. This is just to keep things simple: had we accounted for this, the space usage of the above strategy would increase by only an additive $O(\log n)$.

\section{Warm-up: A Lower Bound for a Blind Player}

We shall now prove that the simple upper bound in \Cref{prop:ub} is tight in the blind player model. In fact, we can prove something stronger. Consider an $(S,T)$-strategy for \mem and let $s := \floor{S/\log n}$ as above. Since each examination of a card reveals the answer to at most $s$ queries of the form ``$x_i = x_j$?'', it suffices to prove that $\Omega(n^2)$ queries must be made by any {\em query strategy} that can access information about $\bx$ only through such pairwise queries. It then follows that $T = \Omega(n^2/s)$, giving $ST = \Omega(n^2 \log n)$.

To analyze such a query strategy, consider the {\em knowledge graph} $G$ of the algorithm, defined as follows: $G$ is an undirected graph with vertices $1, 2, \ldots, 2n$ (the indices of the cards); edge $\{i,j\}$ occurs in $G$ iff $x_i = x_j$ is consistent with the validity of the input $\bx$ and all answers to queries received by the algorithm. This means that, at all times, every edge of $G$ is {\em useful}, where a {\em useless edge} is defined to be one that does not belong to any perfect matching. (In the terminology of matching theory, each connected component of $G$ is $1$-extendable.) Further, as the algorithm proceeds, $G$ can only lose edges. The algorithm can conclude that $x_i = x_j$ iff $\{i,j\}$ is an isolated edge in $G$. It is done when and only when $G$ becomes a perfect matching.

We shall also allow the algorithm the following information for free: we can only have $x_i = x_j$ if exactly one of $i$ and $j$ is in $[n]$. This means that, at the start of the algorithm, the knowledge graph $G$ is a complete bipartite graph with ``left side'' $L = [n]$ and ``right side'' $R = [2n] \setminus [n]$.

\begin{theorem} \label{thm:blind}
  Every query strategy for \mem using pairwise equality queries must make $\Omega(n^2)$ queries in total. Therefore, an $(S,T)$-strategy in the blind player model requires $ST=\Omega(n^2 \log n)$.
\end{theorem}
\begin{proof}
We shall prove our lower bound by an adversarial argument, using the following adversary. Every time the query algorithm asks whether $x_i = x_j$, the adversary answers ``No'' unless consistency requires a ``Yes'' answer. This means that the knowledge graph $G$ evolves as follows, with each query made.
\begin{itemize}
  \item The algorithm asks ``Is $x_i = x_j$?'' for some non-isolated edge $\{i,j\}$ in $G$. The adversary answers ``No'', causing this edge to be {\em deleted} from $G$.
  \item The algorithm considers each remaining edge in $G$ and checks whether it is still useful. The edges which are not useful are said to {\em vanish} from $G$.
\end{itemize}
When the algorithm finishes, $G$ must reduce to a perfect matching: call it $M$. Further, exactly $n(n-1)$ edges must have been either deleted or vanished from $G$. Each query results in exactly one deletion, so it suffices to prove that $\Omega(n^2)$ edges get deleted by the time $G$ becomes $M$. The trouble is that a deletion {\em could} be accompanied by up to $\Omega(n^2)$ vanishings: for instance, consider the deletion of the edge $\{n,n+1\}$ from the bipartite graph whose edge set is $\{\{n,n+1\}\} \cup \{\{i,n+j\}:\, 1 \le i \le j \le n\}$.

The key is to observe that the edges outside $M$ can be partitioned into pairs so that, within each pair, the first of the edges to be removed from $G$ must be deleted, rather than vanished. To formalize this, we first renumber the vertices so that final matching $M$ is $\{\{i,n+i\}:\, 1 \le i \le n\}$. Now, for each $e \notin M$ in the initial complete bipartite graph, define $\phi(e)$ as follows:
\[
  \text{If } e = \{i,n+j\}, \text{ then } \phi(e) := \{j,n+i\} \,.
\]
On the edges outside $M$, this mapping $\phi$ is an involution with no fixed points, so it partitions those edges into pairs, as required.
\medskip

\noindent {\em Claim.~} For all $e \notin M$, either $e$ or $\phi(e)$ was deleted from $G$, and not vanished.

\noindent {\em Proof of Claim.~} Suppose not. Then, either $e$ and $\phi(e)$ were vanished as a result of the same query, or they were vanished as a result of two different queries. In the former case, consider the state of $G$ immediately after the deletion of the edge involved in this query. At this point $G$ contained both $e$ and $\phi(e)$, and so both these edges belonged to the perfect matching $(M \setminus \{\{i,n+i\},\{j,n+j\}\}) \cup \{e,\phi(e)\}$. Therefore both these edges were useful and should not have vanished, a contradiction.

In the latter case, suppose WLOG that $e$ was the first of the two edges to vanish. Consider the state of $G$ immediately after the deletion of the edge involved in the query that caused $e$ to vanish. Arguing as above, $e$ is useful at this point and so it should not have vanished, a contradiction. This proves the claim.
\medskip

Therefore, the total number of edges deleted from $G$ is at least $n(n-1)/2=\Omega(n^2)$, as required.
\end{proof}


\section{The Main Lower Bound} \label{sec:lb}

\subsection{The Overall Setup} \label{sec:lb-setup}

Let $\cB$ be a layered branching program of length $T$ and width at most $2^S$ that correctly plays the Memory Game. Let $r$ and $t$ be positive-integer-valued parameters to be chosen later. We divide $\cB$ into {\em stages}, where each stage, except perhaps the last, consists of $r+1$ consecutive layers of nodes of $\cB$, and the final layer of nodes of stage~$i$ equals the first layer of nodes of stage~$i+1$. Call a stage {\em productive} for an input $\bx \in \cX$ if running $\cB$ on $\bx$ produces at least $2t$ outputs in that stage. Recall, from \Cref{sec:prelim}, that $\cX$ denotes the set of {\em valid} inputs.

Assume, towards a contradiction, that the number of stages is at most $n/(2t)$. Then, for every input $\bx$, there exists a stage productive for $\bx$. Therefore, there exists a stage $i$ productive for a set $\cX'$ of inputs, where $|\cX'| / |\cX| \ge 2t/n$. Consider the first layer of nodes in this stage. There are at most $2^S$ such nodes, so one of these nodes---$v$, say---is reached by at least a $2^{-S}$ fraction of the inputs in $\cX'$.

Consider the subprogram of $\cB$ with source node $v$, consisting of only the nodes in stage $i$. Unfold this subprogram into a decision tree with output (by replicating nodes as needed) and, if necessary, increase its depth to $r$ by querying some dummy variables. Let $\cT$ be the resulting tree. Then $\cT$ produces at least $2t$ correct outputs for a $(2t/n)2^{-S}$ fraction of inputs in $\cX$.

Set $r = \floor{(2/e)\sqrt{nt}}$. The main technical argument in our proof shows that a tree with this depth is ``too shallow'' and is therefore quite unproductive. Specifically, by \Cref{lem:productive-tree}, which we prove below, the fraction of inputs in $\cX$ for which $\cT$ correctly produces at least $2t$ outputs is at most $e^{-t} + (n/2)^{-t}$. Therefore,
\[
  \frac{2t}{n}\cdot 2^{-S} \le e^{-t} + \left(\frac{n}{2}\right)^{-t} \,.
\]
Setting $t = S$ gives us a contraction.

Therefore, for the above choices of $r$ and $t$, the branching program $\cB$ must have more than $n/(2t)$ stages. Since each non-last stage has length $r$, it follows that
\[
  T \ge \frac{nr}{2t} = \frac{n\floor{(2/e)\sqrt{nS}}}{2S} = \frac{\Omega(n^{3/2})}{\sqrt S} \,,
\]
which proves the tradeoff $T\sqrt S = \Omega(n^{3/2})$. This outlines proves our main result.

\begin{theorem}[Main Theorem] \label{thm:main}
  Any deterministic branching program of length $T$ and width at most $2^S$ that correctly plays the Memory Game must obey the tradeoff $ST^2 = \Omega(n^3)$.
\end{theorem}

In the rest of \Cref{sec:lb}, we fill in the necessary details to formally prove \Cref{thm:main}. In the sequel, we indicate how to  generalize the lower bound to randomized strategies.

\subsection{A Probability-Theoretic Lemma}

Our eventual proof of the unproductivity of shallow decision trees crucially hinges on a probability-theoretic lemma that we now develop. Suppose that $r$ elements are chosen, without replacement, from the set $[n] \times \b$, forming a random subset $A$. Define the random variable
\[
  Y = |\{j \in [n]:\, (j,0) \in A \text{ and } (j,1) \in A\}| \,.
\]
Note that $Y = \sum_{j=1}^n I_j$, where $I_j = \indic_{(j,0) \in A \wedge (j,1) \in A}$. These indicator random variables $I_1, \ldots, I_n$ are clearly not independent. Nevertheless, we can prove the following strong tail estimate on $Y$.

\begin{lemma} \label{lem:exp-bound}
  For all $t \ge 1$,
  \[
    \Pr[Y \ge t] \le \exp\left(-t \ln \frac{4nt}{er^2} \right) \,.
  \]
  In particular, for $r \le (2/e)\sqrt{nt}$, we have $\Pr[Y \ge t] \le e^{-t}$.
\end{lemma}
\begin{proof}
  For $1 \le j \le n$, put $U_j = \indic_{(j,0) \in A}$ and $V_j = \indic_{(j,1) \in A}$, so that $I_j = U_j V_j$. We shall use some terminology and results from Joag-Dev and Proschan~\cite{JoagDevP83}. The tuple $(U_1, V_1, \ldots, U_n, V_n)$ has a permutation distribution~\cite[Definition~2.10]{JoagDevP83}. Therefore, these random variables are negatively associated~\cite[Theorem~2.11]{JoagDevP83}. Since the variables $I_j$ are increasing functions of disjoint subsets of these variables, they too are negatively associated~\cite[Property~P$_6$]{JoagDevP83}. Invoking a result from Dubhashi and Panconesi~\cite[Theorem~3.1]{DubhashiP-book}, we can ``apply Chernoff-Hoeffding bounds as is'' to the sum $Y = \sum_{j=1}^n I_j$.
  
  To be precise, note that
  \[
    \EE Y = \sum_{j=1}^n \EE I_j = n \cdot \frac{r}{2n} \cdot \frac{r-1}{2n-1} \le \frac{r^2}{4n} \,.
  \]
  So consider the binomially distributed random variable $Z \sim \Bin(n,p)$, where $p = r^2/(4n^2)$. Then, for all $t > 0$, we have $\Pr[Y \ge t] \le \Pr[Z \ge t]$. We shall use the following very precise Chernoff bound; see, e.g., Theorem~1 of Arratia and Gordon~\cite{ArratiaG89}. For all reals $a, p$, with $0 < p < a < 1$,
  \begin{align}
    \Pr[Z \ge an] &\le e^{-n\, \relent{a}{p}} \,, \label{eq:chernoff} \\
    \text{where } \relent{a}{p} &= a \ln \frac{a}{p} + (1-a) \ln \frac{1-a}{1-p} \label{eq:relent}
  \end{align}
  is the relative entropy of the Bernoulli distribution $\Bern(a)$ to $\Bern(p)$. We simplify~\eqref{eq:relent} as follows.
  \[
    \relent{a}{p} \ge a \ln(a/p) + (1-a) \ln(1-a) \ge a \ln(a/p) - a = a \ln(a/(ep)) \,.
  \]
  Using this in~\eqref{eq:chernoff}, with the setting $a = t/n$, gives us the claimed bound on $\Pr[Y \ge t]$.
\end{proof}

\subsection{The Unproductivity of Shallow Decision Trees}

We now come to the main technical thread of the proof. As shown in \Cref{sec:lb-setup}, proving this lemma will complete the proof of the tradeoff lower bound $ST^2 = \Omega(n^3)$.

\begin{lemma}[Main technical lemma] \label{lem:productive-tree}
Consider a subprogram of $\cB$ of depth $r \le \floor{n/2}$. Then for every $t \le \floor{r/2}$, the probability that a uniformly-random input would produce at least $2t$ outputs following this subprogram is at most $(n-r-t)^{-t}+e^{-t}$.
\end{lemma}
\begin{proof}
Unfold this subprogram into a decision tree with output (by replicating nodes as needed). If necessary, remove any redundant nodes in the tree, where an already-queried variable is re-queried, and increase the depth of every leaf to $r$ by querying some dummy variables. Let $\cT$ be the resulting tree. Then, in $\cT$, along every path from the root to a leaf, exactly $r$ inputs are queried; no input is queried and no output is generated more than once along any path.

Let $\Pi$ be the set of all root-to-leaf paths in $\cT$. Let $\bx$ be a uniformly random input. Let $\pi(\bx)$ denote the path that $\bx$ follows in $\cT$.

For each path $\pi \in \Pi$, let $s(\pi)$ denote the number of outputs along $\pi$. Suppose those $s(\pi)$ outputs are $(y_k^1, y_k^2, v_k)$, for $1 \le k \le s(\pi)$. Recall that each such output is a declaration that $x_{y_k^1} = x_{y_k^2} = v_k$. Let $q_1(\pi), q_2(\pi), \ldots, q_r(\pi)$ be the indices of the positions queried along $\pi$, in the order that they are queried: thus $q_i(\pi) \in [2n]$ for each $i \in [r]$. Suppose the results of the queries are $x_{q_i(\pi)} = w_i(\pi)$ for each $i \in [r]$. Define an ordered list $W(\pi) = [w_1(\pi), w_2(\pi), \dots, w_r(\pi)]$. Then there is a bijection between such an ordered list and a path; i.e. no two distinct paths $\pi_1, \pi_2$ have $W(\pi_1) = W(\pi_2)$ and any ordered list of $r$ numbers picked from $\{1, 1, 2, 2, 3, 3, \dots, R, R \}$ corresponds to a path.

Let $p(\pi) = \Pr{[\pi(\bx) = \pi]}$ denote the probability that a uniformly random input follows path $\pi$. Then
\[
  p(\pi) = \frac{|\{\bx \in \cX:\, x_{q_i(\pi)} = w_i(\pi), \forall i: 1 \le i \le r\}|}{|\cX|}
  = \frac{|\{\bx \in \cX:\, x_{i} = w_i(\pi), \forall i: 1 \le i \le r\}|}{|\cX|} \,,
\]
since all indices are symmetric.

For a uniformly-random input $\bx$, let random variable $X$ be the total number of equal pairs in $W(\pi(\bx))$.

Let $Y$ be defined as in \Cref{lem:exp-bound}, above. Clearly, $Y$ is the number of equal pairs found upon querying $x_1, x_2, \dots, x_r$ for a uniformly random input. We now argue that $X \equiv Y$ in distribution, i.e., that $\Pr[X = u] = \Pr[Y = u]$, for all $u$ with $0 \le u \le \lfloor r/2 \rfloor$.
Let $\Pi_u \subseteq \Pi$ denote the set of all paths in $\cT$ with exactly $u$ equal pairs in their $r$ queries. Then

\begin{align*}
  \Pr{[X = u]} = \Pr{[\pi(\bx) \in \Pi_u]} &= \sum_{\pi \in \Pi_u}{p(\pi)}
  = \frac{\sum_{\pi \in \Pi_u}|\{\bx \in \cX:\, x_{i} = w_i(\pi), \forall i: 1 \le i \le r\}|}{|\cX|} =\Pr{[Y = u]} \,.
\end{align*}

Since $X \equiv Y$, it follows from \Cref{lem:exp-bound} that
\[
  \Pr{[X \ge t]} = \Pr{[Y \ge t]}\le e^{-t} \,.
\]

Call a path $\pi$ in $\cT$ \textit{good} if and only if it produces at least $2t$ outputs and at least $t$ of those outputs are of two variables which have both been queried along $\pi$. Call an input $\bx$ \textit{good} if and only if $\pi(\bx)$ is good. So the probability that a uniformly-random input is good satisfies $$\Pr{[\pi(\bx) \text{ is good}]} \le \Pr{[X \ge t]} \le e^{-t}.$$

When an input passes through a bad path $\pi$, either $\pi$ does not generate enough outputs, or it generates at least $t+1$ outputs of two variables that the path does not both query. Pick $t+1$ such outputs arbitrarily. Among those $t+1$ output pairs, suppose $k$ pairs have exactly one variable queried along $\pi$, and $p=t+1-k$ pairs have neither variables queried along $\pi$.

We first consider the $k$ outputs with exactly one variable queried. Denote those $k$ queried variables as $x_{b_1}, x_{b_2}, \dots, x_{b_k}$. In total, there are at least $n-r$ variables that $\pi$ does not query. For a uniformly random input consistent with queries along $\pi$, all of those unqueried variables are equally likely to match with $x_{b_1}$. Therefore, the probability that the variable matched with $x_{b_1}$ is correct is at most $1/(n-r)$. For $1 \le i \le k-1$, suppose the variables matched with $x_{b_1}, \dots, x_{b_i}$ are correct, we have at least $n-r-i$ equally likely choices to match with $x_{b_{i+1}}$. Therefore, the probability that $x_{b_1}, x_{b_2}, \dots, x_{b_k}$ are all matched correctly is at most
\[
  \left(\frac{1}{n-r}\right)\left(\frac{1}{n-r-1}\right) \cdots \left(\frac{1}{n-r-k+1}\right) \le (n-r-k+1)^{-k} \,.
\]

Next we consider the $p$ outputs with neither variable queried. Denote the {\em values} of those $p$ pairs as $v_{c_1}, v_{c_2}, \dots, v_{c_p}$. In total, there are at least $R-r$ possible values of variables that $\pi$ does not discover in its queries. For a uniformly random input consistent with queries along $\pi$, all of those values of variables are equally likely to be the value of any of the $p$ pairs. Therefore, the probability that $v_{c_1}, v_{c_2}, \dots, v_{c_p}$ are all correct values is at most
\[
  \left(\frac{1}{R-r}\right)\left(\frac{1}{R-r-1}\right) \cdots \left(\frac{1}{R-r-p+1}\right) \le (n-r-p+1)^{-p} \,.
\]

Therefore, if we uniformly randomly draw one input $\bx$ from the set of all bad inputs, the probability that $\bx$ generates at least $2t$ correct outputs along $\pi(\bx)$ satisfies
\begin{align*}
  \Pr{[\bx \text{ produces at least } 2t \text{ correct outputs} \mid \bx \text{ is bad}]} 
  &\le \Pr{[\text{Guesses of } k+p \text{ matches are all correct}] } \\
  &\le (n-r-k+1)^{-k} (n-r-p+1)^{-p} \\
  &\le \left(n-r-t\right)^{-t-1} \\
  &< (n-r-t)^{-t} \,.
\end{align*}

Therefore, the total probability of a uniformly random input $\bx$ (either good or bad) producing at least $2t$ correct outputs satisfies
\begin{align*}
  \Pr{[\bx \text{ produces at least } 2t \text{ correct outputs}]} 
  &< (n-r-t)^{-t}\Pr{[\bx \text{ is bad}]} + \Pr{[\bx \text{ is good}]} \\
  &< (n-r-t)^{-t}+e^{-t} \,,
\end{align*}
which completes the proof.
\end{proof}

\subsection{Lower Bound for Randomized Algorithms}

The lower bound proved so far applies only to {\em deterministic} strategies for the Memory game. However, the ideas readily extend to give a similar lower bound for randomized strategies. We sketch this generalization.

Given space constraint $S$, suppose a Las Vegas algorithm $\mathcal{A}$ for the Memory game takes expected time $T = T(n,S)$. Convert $\mathcal{A}$ to a Monte Carlo algorithm $\mathcal{A'}$ that exactly mimics $\mathcal{A}$ expect that it always stops after $10T$ steps, at which point it outputs "Error". By Markov's Inequality, $\mathcal{A'}$ is correct for any input with probability $9/10$.

By Yao's Lemma, there is a deterministic algorithm running in time $10T$ that generates $n$ correct outputs for a set $\cX''$ of inputs, where $|\cX''| / |\cX| \ge 9/10$. Model this deterministic algorithm as a branching program $\cB$ with length $10T$ and width at most $2^S$.

Similar to the proof for the deterministic lower bound, we divide $\cB$ into stages of length $r = \floor{(2/e)\sqrt{nS}}$. Assume that the number of stages is at most $n/(2S)$. Then there exists a subtree at some stage that produces at least $2S$ correct outputs for a $(2S/n)2^{-S}$ fraction of inputs in $\cX''$. We arrive at the contradiction that
\[
  \frac{9}{10} \cdot \frac{2S}{n}\cdot 2^{-S} \le e^{-S} + \left(\frac{n}{2}\right)^{-S} \,.
\]

So $\cB$ has more than $n/(2S)$ stages. Since each non-last stage has length $r$, it follows that
\[
  10T \ge \frac{nr}{2S} = \frac{n\floor{(2/e)\sqrt{nS}}}{2S} = \frac{\Omega(n^{3/2})}{\sqrt S} \,,
\]
which proves the tradeoff $T\sqrt S = \Omega(n^{3/2})$.


\section{Concluding Remarks} \label{sec:concl}

In this work, we studied the complexity of the Memory Game (also known as ``Concentration'') in a limited-memory setting, proving a nontrivial time-space lower bound. We showed that, when the player has $S$ bits of memory, the number of card flips required to guarantee completion of the game is $\Omega(n^{3/2}/\sqrt{S})$, for both deterministic and randomized algorithms.

Our results suggest a number of directions for future work.

First, we could hope to obtain a stronger bound in a comparison-based branching program model. In this more restricted model, each node in the branching program compares two variables $x_i, x_j$ and branches three ways, corresponding to the three cases $x_i<x_j$, $x_i=x_j$, and $ x_i>x_j$, respectively. We conjecture that a stronger bound $ST = \Omega(n^{2-\epsilon})$ is achievable in this model.

We further conjecture that every randomized query strategy for \mem using pairwise equality queries must make $\tOmega(n^2)$ queries in total, i.e., that our \Cref{thm:blind} extends to the randomized case.

This work was motivated in part by the general question of what time-space tradeoffs one can obtain for the problem (in fact, family of problems) that calls for finding a large matching in an input graph. There are a number of interesting space-bounded algorithms for approximate maximum matching in the data streaming model; see, e.g., Crouch and Stubbs~\cite{CrouchS14} and the references therein. There are also some corresponding lower bounds for this problem in a one-pass streaming model; see, e.g., Goel et al.~\cite{GoelKK12} and Assadi et al.~\cite{AssadiKL17}. The streaming model is a very restrictive model for space-bounded algorithms: it requires the input to be accessed in a sequential fashion. One-pass streaming is even more restrictive. Yet, the precise tradeoff between space and approximation quality for maximum matching remains open even in the one-pass case, and there are essentially no nontrivial lower bounds in the multi-pass case. We hope that the lessons learned in studying the Memory game will find applications in establishing time-space trade-offs for finding large matchings in graphs. One can think of the Memory game as a graph on $2n$ vertices where only equal pairs are connected by an edge. What we have shown here are tradeoffs for discovering this perfect matching that is promised to exist.

\section*{Acknowledgments}

We would like to thank Graham Cormode and Peter Winkler for helpful discussions about the problems considered in this work.

\bibliographystyle{alpha}
\bibliography{refs}

\appendix
\appendix
\section{A Lower Bound for the Unique Pairs Problem}
\label{unique2}
In section \ref{RelatedWork}, we compared the hardness of \mem to \textsc{unique-elements}. Here we show that a time-space tradeoff of $TS=\Omega(n^2)$ can be obtained for a variant of \mem that is a harder task (and thus easier to prove a lower bound). 

Consider the following problem of \textsc{unique-pairs}: Given an input of $2n$ integers in $[n]$, output all pairs $(i,j)$ with $i < j$ such that $x_i = x_j$ and no other $k \in [2n]$ satisfies $x_i = x_j = x_k$. Modeling the proof for \textsc{unique-elements} in \cite{BeameLecNotes}, we now prove that any algorithm computing \textsc{unique-pairs} that runs in time at most $T$ and uses space at most $S$ has $TS=\Omega(n^2)$.

\begin{proof}
For a uniformly random input over $[n]^{2n}$, the expected output size for \textsc{unique-pairs} is $$\EE{[\text{\# outputs}]}=\binom{n}{2} \frac{1}{n} \left(1-\frac{1}{n}\right)^{2n-2}>\frac{n-1}{2e^2}$$
By Markov's inequality, the probability that the output size is at least $(n-1)/(4e^2)$ is $$\Pr{\left[\text{\# outputs}>\frac{n-1}{4e^2}\right]} \ge c$$ for some constant $0<c<1$.

Let $\cB$ be the layered branching program computing \textsc{unique-pairs}. We have $T \ge 2n$. For some $h$ to be chosen later, we partition $\cB$ into $T/h$ stages of depth $h$. We call an input \textit{good} if \textsc{unique-pairs} generates at least $(n-1)/(4e^2)$ outputs for that input. So $\cB$ run on a good input should produce at least $$m=\frac{n-1}{4e^2(T/h)}=\frac{(n-1)h}{4e^2T}$$ outputs in some stage. Same as the argument in \cite{BeameLecNotes}, for $m \le n/4$ and $h \le n/4$, for any $\cB$ of depth at most $h$, $$\Pr[\cB \text{ produces at least } m \text{ correct outputs}]\le e^{-c'm}$$ for some constant $c'>0$. Choosing $h=n/4$, we have $m \le n/4$. Since the fraction of good inputs is at least $c$, we need $$2^S e^{-c'm} \ge c$$ which gives us $S=\Omega(m)=\Omega(n^2/T)$.
\end{proof}

\end{document}